\documentclass{article}
\usepackage{amsmath,amsthm,amssymb}

\def\0{\mathbf{0}}
\def\1{\mathbf{1}}
\DeclareMathOperator*{\argmax}{arg\,max}

\def\Fix{\mathrm{Fix}}

\def\R{\mathbb{R}}

\newtheorem{thm}{Theorem}[section]

\newtheorem{lm}[thm]{Lemma}
\newtheorem{cor}[thm]{Corollary}

\newtheorem{ft}[thm]{Fact}
\theoremstyle{remark}
\newtheorem{rk}[thm]{Remark}
\theoremstyle{definition}
\newtheorem{df}[thm]{Definition}
\newtheorem{eg}[thm]{Example}

\usepackage[backend=bibtex, style=alphabetic, maxbibnames=99]{biblatex} 
\DeclareFieldFormat{postnote}{\mknormrange{#1}}
\DeclareFieldFormat{volcitepages}{\mknormrange{#1}}
\DeclareFieldFormat{multipostnote}{\mknormrange{#1}}

\usepackage{enumitem}
\usepackage[hidelinks]{hyperref}

\usepackage{multirow,array}
\usepackage{orcidlink}
\title{Order-theoretical fixed point theorems for correspondences and application in game theory}
\author{Lu Yu\thanks{Université Paris 1 Panthéon-Sorbonne, UMR 8074, Centre d'Economie de la Sorbonne, Paris, France, \href{mailto:yulumaths@gmail.com}{yulumaths@gmail.com}} \,\orcidlink{0000-0001-6154-4229}}
\date{\today}
\begin{document}
	\maketitle
	\begin{abstract}For an ascending correspondence $F:X\to 2^X$ with chain-complete values on a complete lattice $X$, we prove that the set of fixed points is a  complete lattice. This strengthens Zhou's fixed point theorem.  For chain-complete posets that are not necessarily lattices, we generalize the Abian-Brown and the Markowsky fixed point theorems from single-valued maps to multivalued correspondences. {We provide an application in game theory.}\end{abstract}

\section{Introduction}
Fixed point theorems play a fundamental role in the study of existence of Nash equilibria of games. For example, Nash \cite{nash1950equilibrium} uses Brouwer's fixed point theorem to prove the existence of Nash equilibria for normal form games. Shapley \cite[p.1097]{shapley1953stochastic} applies Banach's fixed point theorem to show that  value exists for zero-sum stochastic games. Another pivotal fixed point theorem (Fact \ref{ft:Tarski}) of order nature due to Tarski \cite{tarski1955lattice} is utilized by  Topkis \cite{topkis1979equilibrium} to investigate supermodular games.   We shall review several results subsequent to Tarski's result and give some generalizations.

\subsection*{Notation}For a self-correspondence $F:S\to 2^S$ on a set $S$, let  $\Fix(F):=\{x\in S|x\in F(x)\}$ denote the set of fixed points of $F$.  A set with a partial order is called a poset. For a poset $X$, let $\0$ (resp. $\1$) denote the minimum (resp. maximum) element whenever it exists. For   subsets $S$ and $T$ of $X$, we write $S\le T$ (resp. $S<T$) to mean that  $s\le t$ (resp. $s<t$)  for all $s\in S$ and $t\in T$. For every self-correspondence $F:X\to 2^X$ and every $h\in X$, define a ``truncated'' correspondence $F^h:[h,\infty)\to 2^{[h,\infty)}$ by $x \mapsto F(x)\cap [h,\infty)$. Then  $\Fix(F^h)=\Fix(F)\cap [h,\infty)$. If $\0=\min X$ exits, then $F^{\0}=F$. Let $A_F:=\{x\in X|F(x)\cap [x,+\infty)\neq\emptyset\}$ and $B_F:=\{x\in X|F(x)\cap(-\infty,x]\neq\emptyset\}$.

\subsection*{Fixed point theorems on lattices}
We recall some known fixed point theorems for maps and correspondences defined on lattices, and state our first main result.

A poset is a \emph{complete lattice}, if every nonempty subset has a least upper bound as well as a largest lower bound. Let $X$ be a nonempty complete lattice. 
\begin{ft}[Knaster-Tarski, {\cite[Thm.~1]{tarski1955lattice}}]\label{ft:Tarski}
Let $f:X\to X$ be an increasing map. Then $\Fix(f)$ is a nonempty complete lattice. Its largest (resp. least) element is $\sup_X(A_f)$ (resp, $\inf_X(B_f)$).
\end{ft}
Topkis applies Fact \ref{ft:Tarski} to get the existence of a largest and a least equilibria of  supermodular games in non-product spaces in \cite[Theorem 3.1]{topkis1979equilibrium}.
Finer study of Nash equilibria requires a multivalued generalization of Fact \ref{ft:Tarski}, namely Fact  \ref{ft:VeinottZhou}. To state it, we need  a notion of increasingness for correspondences.  Different candidate definitions exist in the literature.  All the multivalued monotone properties introduced below reduce to the usual increasing condition for single-valued  maps.  Definition \ref{df:Zasc} is introduced by Veinott  (see, e.g., \cite[p.308]{topkis1978minimizing} and \cite[p.296]{zhou1994set}).
\begin{df}\label{df:Zasc}
Let $P$ be a poset. Let $Y$ be a lattice. Let $F:P\to 2^Y$ be a correspondence. If for any $x\le x'$ in $P$, every $y\in F(x)$ and every $y'\in F(x')$, we have $y\vee y'\in F(x')$ and $y\wedge y'\in F(x)$,  then $F$ is called  \emph{ascending}.
\end{df}If $F:P\to 2^Y$ is ascending, then $F(x)$ is a sublattice of $Y$ for every $x\in P$.

Fact \ref{ft:VeinottZhou} appears  in Zhou's work  \cite[Theorem 1]{zhou1994set}. It is a special case of  Veinott's unpublished work \cite[Ch. 4, Theorem 14]{veinott1992lattice}.   
\begin{ft}[Veinott-Zhou]\label{ft:VeinottZhou}
Let $F:X\to 2^X$ be an ascending correspondence. If  $F(x)$ is a nonempty \emph{subcomplete} sublattice of $S$ for every $x\in X$,   then $\Fix(F)$ is a nonempty complete lattice.
\end{ft}For a parameterized game with complementarities, Veinott \cite[Ch. 10, Theorem 2]{veinott1992lattice} applies Fact \ref{ft:VeinottZhou} to prove that the set of Nash equilibria is  a nonempty complete lattice.   Based on Fact \ref{ft:VeinottZhou}, Zhou \cite[Theorem 2]{zhou1994set}  also proves that the set of Nash equilibria is  a nonempty complete lattice for  a supermodular game in product space.

Using Tarski's original argument in \cite{tarski1955lattice}, Calciano \cite{calciano2010theory} gives a generalization of Fact \ref{ft:VeinottZhou}. It involves  increasingness conditions  weaker than  ascendingness.
\begin{df}
Let $S, Y$ be two posets, and let $F:S\to 2^Y$ be a correspondence. \begin{enumerate}
	\item (\cite[Definitions 9]{calciano2010theory}) If for any $x\le x'$ in $S$ and every $y\in F(x)$ (resp. every $y'\in F(x')$), there is $y'\in F(x')$ (resp. $y\in F(x)$) with $y'\ge y$, then $F$ is  \emph{upper} (resp. \emph{lower}) \emph{increasing}.
	\item (\cite[Definition 10]{calciano2010theory}) Assume that $Y$ is a lattice. If for any $x\le x'$ in $S$, every $y\in F(x)$ and every $y'\in F(x')$, there is $q\in F(x')$ (resp. $p\in F(x)$) such that $q\in [y,y\vee y']$ (resp. $p\in [y\wedge y',y']$), then $F$ is  \emph{strongly upper} (resp. \emph{lower}) \emph{increasing}.
\end{enumerate}
\end{df}
Every ascending correspondence is strongly upper increasing and strongly lower increasing. Every strongly upper (resp. lower) increasing correspondence of nonempty values is upper (resp. lower) increasing.

\begin{thm}[Calciano, {\cite[Theorem 13]{calciano2010theory}}]\label{thm:CalZhou}
 Let $F:X\to 2^X$ be a correspondence of nonempty values.   Assume that \begin{enumerate}
	\item\label{it:Fhasmax} for every $x\in X$, the value $F(x)$ has a greatest element;
	\item\label{it:Fuppinc} $F$ is upper increasing;
	\item\label{it:Fstr} $F$ is strongly lower increasing;
	\item\label{it:Fhmin} for every $h\in A_F$ and every $x\in [h,\1]$, the value $F^h(x)$ has a least element whenever nonempty.
\end{enumerate} Then  $\Fix(F)$ is a nonempty complete lattice.
\end{thm} 
As an application of Theorem \ref{thm:CalZhou}, Calciano \cite[Theorems 24 and 26]{calciano2010theory} proves that the set of Nash equilibria of a uniformly g-modular game (in the sense of \cite[Definition 21]{calciano2010theory}) is a nonempty complete lattice. 

Sabarwal's result \cite[Theorem 4]{sabarwal2023general} has a flavor similar to that of Theorem \ref{thm:CalZhou}. We give new proofs of Theorems \ref{thm:CalZhou} and \ref{thm:Sabarwal} in Section \ref{sec:2pf}.\begin{thm}[Sabarwal]\label{thm:Sabarwal}Let $F:X\to 2^X$ be a correspondence of nonempty values. Assume that
	\begin{enumerate}[label=\emph{\Roman*})]\item\label{it:isotonesup} for every $x\in X$, the element $\max F(x)$ exists and $\max F:X\to X$ is increasing.
		\item\label{it:isoinfupint} for every $h\in A_F$ and  every $x\ge h$ in $X$, the value $F^h(x)$ has a least element and $\min F^h:[h,\1]\to [h,\1]$ is increasing. \end{enumerate} Then $\Fix(F)$ is a nonempty complete lattice. 
\end{thm}
Fact \ref{ft:weakZhou} weakens ascendingness in Fact \ref{ft:VeinottZhou} to V-ascendingness (Definition \ref{df:Vascending}), and relax the subcompleteness in two directions.
\begin{ft}[{\cite[Theorems 1.2 and 1.3]{yu2023generalization1}}]\label{ft:weakZhou}
	Let $F:X\to 2^X$ be a V-ascending correspondence with nonempty values. Suppose that 
	\begin{itemize}
		\item either $F(x)$ is chain-subcomplete in $X$ for every $x\in X$, 
		\item or $F(x)$ is a complete lattice  for every $x\in X$.
	\end{itemize} Then $\Fix(F)$ is a nonempty complete lattice.
\end{ft}Fact \ref{ft:weakZhou} can be used to give a purely order-theoretic generalization \cite[Theorem 1.4]{yu2023generalization1} of Topkis's theorem for supermodular games.

Our first main result Theorem \ref{thm:chaincpltvalue} assumes a hypothesis weaker than both conditions in Fact \ref{ft:weakZhou}.  Example \ref{eg:Zhou>wkZhou} shows that Theorem \ref{thm:chaincpltvalue} is a strict generalization of Fact \ref{ft:weakZhou}.
\begin{thm}\label{thm:chaincpltvalue}
	Let $F:X\to 2^X$ be a V-ascending correspondence. Suppose that for every $x\in X$, the value  $F(x)$ is nonempty, chain-complete downwards and chain-bounded above. Then $\Fix(F)$ is a nonempty complete lattice.
\end{thm} Another main result, Theorem \ref{thm:myZhou2}, also generalizes Fact \ref{ft:VeinottZhou}. Compared to Theorem \ref{thm:chaincpltvalue}, it weaken the V-ascending condition further, but requires chain-subcomplete values. Section \ref{sec:game} gives an application of Theorem \ref{thm:myZhou2}, which analyzes Nash equilibria for an extension of quasi-supermodular games. 
\subsection*{Fixed point theorems for general order}
We collect several fixed point theorems for maps on posets or pseudo-ordered sets.

For posets, Abian and Brown \cite[Thm.~2]{abian1961theorem} relax the completeness hypothesis in Fact \ref{ft:Tarski}. It is an existence result, which does not describe the order structure of the fixed points.
\begin{ft}\label{ft:AbianBrown}
	Let $X$ be a nonempty poset where  every nonempty well-ordered subset  has a least upper bound. Let $f:X\to X$ be an increasing map. Assume that there is $a\in X$  with $a\le f(a)$. Then $f$ has a fixed point.  
\end{ft}
A nonempty poset is called \emph{chain-complete} in \cite[p.53]{markowsky1976chain}, if it is chain-complete upwards (Definition \ref{df:chaincompleteZhou}) and has a least element.\footnote{Some authors call a nonempty poset chain-complete if it is chain-complete both upwards and downwards.} Complete lattices are chain-complete.
Markowsky \cite[Thm.~9]{markowsky1976chain} gives an analog of Fact \ref{ft:Tarski} for chain-complete posets.
\begin{ft}\label{ft:Markowsky}Let $f:X\to X$ be an increasing map on a poset.
	If $X$ is  chain-complete,   then so is $\Fix(f)$.
\end{ft}Markowsky \cite[Cor., p.66]{markowsky1976chain} deduces Fact \ref{ft:Tarski} from Fact \ref{ft:Markowsky}.
The existence of the least element of $X$ in Fact \ref{ft:Markowsky}  is necessary. For example, let $X=\{a,b\}$ be an anti-chain with  two incomparable elements. Then $X$ is chain-complete upwards and downwards. Let $f:X\to X$ be the transposition. Then $f$ is increasing, but has no fixed points.

Our principal results of Section \ref{sec:ABM} generalize Facts \ref{ft:AbianBrown} and \ref{ft:Markowsky} from single-valued maps to multivalued correspondences.

Bhatta \cite[Theorem]{bhatta2005weak} studies the existence of least fixed point of increasing maps defined on  pseudo-ordered sets, relaxing the transitivity in the definition of partial order. Then Bhatta and George \cite[Thm.~3.4]{parameshwara2011some} extend Fact \ref{ft:Markowsky} to pseudo-ordered sets. 

\section{Generalization of Zhou's fixed point theorem}\label{sec:Zhou}
We state our second main result, Theorem \ref{thm:myZhou2}, which is a generalization of Fact \ref{ft:VeinottZhou}. It relies on some conceptions of multivalued increasingness. 
\begin{df}\cite[Ch.~4, Sec.~3]{veinott1992lattice}\label{df:Vascending} Let $X$ be a poset, and let $Y$ be a lattice. Let $F:X\to 2^Y$ be a correspondence. If for any  $x<x'$ in $X$,\footnote{Note the strict inequality here.} every $y\in F(x)$ and every $y'\in F(x')$,  we have 	$y\wedge y'\in F(x)$ (resp. $y\vee y'\in F(x')$), then $F$ is called \emph{lower }(resp. \emph{upper}) \emph{V-ascending}. If $F$ is both upper and lower  V-ascending, then $F$ is called \emph{V-ascending}. 
\end{df}
\begin{rk}An ascending correspondence is V-ascending. In Definition \ref{df:Vascending}, if  we have either $y\wedge y'\in F(x)$ or $y\vee y'\in F(x')$, then $F$ is called weakly ascending.   Upper/lower V-ascending correspondence is weakly ascending.  	Kukushkin \cite[Theorem 2.2]{kukushkin2013increasing} gives a sufficient condition for a weakly ascending correspondence to admit increasing selections.\end{rk}
\begin{df}\label{df:Cascending}
	Let $X,Y$ be two posets. Let $F:X\to 2^Y$ be a correspondence. If for any  $x<x'$ in $X$, every $y\in F(x)$ and every $y'\in F(x')$,  there is $u\in F(x')$ (resp. $v\in F(x)$) with $u\ge y'$ and $u\ge y$ (resp. $v\le y$ and $v\le y'$), then $F$ is called \emph{upper} (resp. \emph{lower}) \emph{C-ascending}. If $F$ is both upper and lower C-ascending, then $F$ is called \emph{C-ascending}.
\end{df}

\begin{rk}
If $F:X\to 2^Y$ is upper C-ascending  with nonempty values, then $F$ is upper increasing for the following reason.	For any $x<x'$ in $X$ and every $y\in F(x)$, take $y'\in F(x')$. Then there is $u\in F(x')$ with $u\ge y'\vee y\ge y$. 
\end{rk}
As Example \ref{eg:VneqC} shows, C-ascending condition is strictly weaker than V-ascending condition.
\begin{eg}\label{eg:VneqC}
	Let $X=\{0,1\}$ be a chain in $\R$. Let $Y=\{0,a,b,1/2,1\}$ be a lattice, with $0=\min Y$, $1=\max Y$, $a,b$ incomparable and $a\vee b=1/2$. Define $F:X\to 2^Y$ by $F(0)=\{0,a\}$, $F(1)=\{b,1\}$. Then $F$ is C-ascending but not upper V-ascending. In fact, one has $a\in F(0)$, $b\in F(1)$, but $a\vee b\notin F(1)$.
\end{eg}
\begin{df}\label{df:chaincompleteZhou}Let $X$ be a nonempty poset.
	\begin{enumerate} \item  If for every nonempty chain $C$  in $X$, there is $x\in X$ with $C\le x$ (resp. $x\le C$), then $X$ is called  \emph{chain-bounded above} (resp. \emph{below});
		\item  If for every nonempty chain $C$  in $X$, the element $\sup_XC$ (resp. $\inf_XC$) exists, then $X$ is  \emph{chain-complete upwards} (resp. \emph{downwards}).  
	\item A subset $S$ of $X$ is  \emph{chain-subcomplete upwards} (resp. \emph{downwards}) in $X$, if for every nonempty chain $C$ in $S$, the element $\sup_XC$ (resp. $\inf_XC$) exists in $S$. \end{enumerate} \end{df}
A chain-complete upwards (resp. downwards) poset  is chain-bounded above (resp. below).  

Theorem \ref{thm:myZhou2} generalizes Fact \ref{ft:VeinottZhou} in three aspects. First, the values $F(x)$ are assumed to be chain-bounded  above and chain-subcomplete downwards rather than to be subcomplete. Second, we do not require the values  $F(x)$ to be  sublattices of $X$. Third,  V-ascending condition is strictly weaker than ascending condition, as Example \ref{eg:01abZhou} shows.
\begin{thm}\label{thm:myZhou2}
	Let $X$ be a nonempty complete lattice. Let  $F:X\to 2^X$ be a correspondence. Assume that \begin{enumerate}[label=\emph{\alph*})]
		\item\label{it:valuehalfsubcpltZhou} for every $x\in X$, the value $F(x)$ is  nonempty, chain-subcomplete downwards (resp.  upwards)  in $X$; 
		\item\label{it:Fischainbddabove} for every $x\in X$, the set $F(x)$ is chain-bounded above (resp. below); 
		\item the correspondence $F$ is lower (resp. upper) V-ascending and upper (resp. lower) C-ascending. 
	\end{enumerate}
	Then $\Fix(F)$ is a nonempty complete lattice. 
\end{thm}\begin{rk}The set $\Fix(f)$ in Fact \ref{ft:Tarski} and Theorem \ref{thm:myZhou2} may not be a sublattice of $X$ even if $X$ is finite.  For instance,	let $X=\{0,1,2\}\times \{0,1\}$, which is  a complete sublattice of $\R^2$. Define $f:X\to X$ by \[
	f(x)=\begin{cases}
		(2,1),&x\in \{(1,1),(2,0),(2,1)\},\\
		x,&\text{else.}
	\end{cases}
	\] Then $f$ is increasing. The set $\Fix(f)$ contains $(0,1)$ and $(1,0)$ but not their join in $X$, namely $(1,1)$. In particular, $\Fix(f)$ is not a join sublattice of $X$. Another example with an infinite $X$  is in \cite[Example 2.5.1]{topkis1998supermodularity}. 
\end{rk}
The proof of Theorem \ref{thm:myZhou2} is divided into two steps. The first step about the existence of extremal fixed points is detailed in Section \ref{sec:extremal}. Next, we apply the existence result to various ``truncations'' of $F$. To show that such a ``truncation'' have nonempty values, two different approaches are provided. The first approach explained in Section \ref{sec:1pf} is in the spirit of \cite[Ch.2, Lemma 8]{veinott1992lattice}. The second approach in Section \ref{sec:2pf} is similar to that in \cite[Ch.~2]{calciano2010theory}. 
\begin{eg}
	Let $S=[0,2]$ be a complete chain. Define a correspondence $F:S\to 2^S$ by $F(s)\equiv[0,1)\cup\{2\}$. Then $F$ is ascending. However,  $F(s)$ is not chain-subcomplete in $S$ for all $s\in S$. So, we cannot apply Fact \ref{ft:VeinottZhou} in this case. But Theorem \ref{thm:myZhou2} guarantees that the fixed points form a complete lattice. Therefore, Theorem \ref{thm:myZhou2} is a proper generalization of Fact \ref{ft:VeinottZhou}.
\end{eg}
\begin{eg}\label{eg:01abZhou}
	Let $X=\{\0,\1,a,b\}$ be a finite lattice, where $\0=\min X$, $\1=\max X$ and $a,b$ are incomparable. Define a correspondence $F:X\to 2^X$ by $F(\0)=\{\0\}$, $F(\1)=\{\1\}$ and $F(a)=F(b)=\{a,b\}$. Then $F$ is V-ascending but not ascending. Indeed, $F(a)$ is not a sublattice of $X$.
Moreover, $F(a)$ does not have any greatest element.  So, we cannot apply Theorem \ref{thm:CalZhou} directly. But Theorem \ref{thm:myZhou2} is applicable in this situation, so it is not covered by Theorem \ref{thm:CalZhou}.	\end{eg}
\begin{eg}
	For a  two-player  normal form game $(X_i,u_i)_{i=1,2}$, let $X=X_1\times X_2$. The best reply correspondence of player $i$ is denoted by $R_i:X_{-i}\to 2^{X_i}$. The joint best reply correspondence is \[R:X\to 2^X,\quad x\mapsto R_1(x_2)\times R_2(x_1).\] Then $\Fix(R)$ is the set of Nash equilibria. We give two examples where  $R$ satisfies the conditions of Theorem \ref{thm:myZhou2}. 
\begin{enumerate}\item Let $C=\{0\}\cup[1,2]$, which is a complete chain. Let $X_1=X_2=C$. Define  payoff functions $u_i:X\to \R$ ($i=1,2$) by \[u_1(x)=\begin{cases}
	0, &x=(1,0)\text{ or }(1,1),\\
	1,&\text{else,}
\end{cases},\quad u_2(x)=\begin{cases}
0,&x=(0,1)\text{ or }(1,1),\\
1,&\text{else.}
\end{cases}\]This defines a symmetric normal form game. For $i=1,2$,  one has \[R_i:C\to 2^C,\quad t\mapsto \begin{cases}
C\setminus\{1\}, & t\le 1,\\
C, &t>1.
\end{cases}\]The value $R(0,0)$ is not chain-subcomplete downwards in $X$, so one cannot apply Fact \ref{ft:VeinottZhou} to $R$. Let $h=(1,0)$. Then $R^h(h)=(1,2]\times (C\times \{1\})$ has no least element. Hence, one cannot apply Theorem \ref{thm:CalZhou} neither. 
\item Let $X_1=\{0, a, b,1 \}$ be as in Example \ref{eg:01abZhou}. 
Let $X_2=\{0, 1\}$. The payoff functions  are shown in the table below. The columns  (resp. rows) represent the strategies of Player $1$ (resp. $2$). In the table, the first (resp. second) number in each pair indicates the payoff for Player $1$ (resp. $2$).

\begin{table}[!h]
	\setlength{\extrarowheight}{2pt}
	\begin{tabular}{cc|c|c|c|c|}
		& \multicolumn{1}{c}{} & \multicolumn{2}{c}{Player $1$}&\multicolumn{1}{c}{}&\multicolumn{1}{c}{}\\
		& \multicolumn{1}{c}{} & \multicolumn{1}{c}{$0$}  & \multicolumn{1}{c}{$1$}&  \multicolumn{1}{c}{$a$} & \multicolumn{1}{c}{$b$} \\\cline{3-6}
		\multirow{2}*{Player $2$}  & $0$ & $(1,1)$ & $(1,1)$&$(0,1)$&$(0,1)$ \\\cline{3-6}
		& $1$ & $(0,0)$ & $(1,1)$&$(1,0)$&$(1,0)$\\\cline{3-6}
	\end{tabular}
\end{table} 

  From the table, one has \[R_1(0)=\{0,1\},\, R_1(1)=\{a,b,1 \},\, R_2(0)=R_2(a)=R_2(b)=0,\, R_2(1)=\{0,1\}.\]

The value  $R(0,1)$ is not a sublattice of $X$, so  Fact \ref{ft:VeinottZhou} is not applicable in this case. Take $h=(1,0)\in X$.  Then $R^h(0,1)$ is nonempty but does not admit a least element. Hence,  Theorem \ref{thm:CalZhou} is  not applicable neither.\end{enumerate} 
\end{eg}

\section{Existence of extremal fixed point}\label{sec:extremal}
We give several results concerning the existence of the largest/least fixed point.
\begin{thm}\label{thm:myCalciano}
	Let	$X$ be a nonempty poset. Let $F,M:X\to 2^X$ be two correspondences. Suppose that \begin{enumerate}[label=\emph{\Alph*})]
		\item \label{it:MinF}for every $x\in X$, the value $F(x)$ is nonempty and $M(x)\subset F(x)$;
		\item \label{it:Mmax}for every $x\in X$ and every $y\in F(x)$, there is $z\in M(x)$ with $y\le z$ (resp. $y\ge z$);
		\item  \label{it:Minc}for any $x<x'$ in $X$, every $y\in M(x)$ and every $y'\in M(x')$, we have $y\le y'$;
		\item\label{it:Anonempty} the subset $A_F$ (resp. $B_F$) is nonempty;
		\item\label{it:x*} the element $x^*:=\sup_X(A_F)$ (resp. $x_*:=\inf_X(B_F)$) exists.
	\end{enumerate}
	Then $x^*$ (resp. $x_*$) is the largest (resp. least) fixed point of $F$.
\end{thm}
\begin{proof}
	By symmetry, it suffices to prove the result without parentheses. By Assumptions \ref{it:MinF} and \ref{it:Mmax}, one has $A_F=A_M$. Then  by Assumption \ref{it:Minc} and  Lemma \ref{lm:7}, one has	$x^*\in A_M$. So there is $y^*\in M(x^*)$ with $x^*\le y^*$.  We claim that $x^*$ is a fixed point of $M$. Assume the contrary, so $x^*<y^*$. By Assumption \ref{it:Minc},  for every $z\in M(y^*)$, we have $y^*\le z$, and whence $y^*\in A$. Then $y^*\le x^*$, a contradiction.  From Assumption \ref{it:MinF} and the claim, $x^*$ is a fixed point of $F$.
	
	We check that $x^*$ is the largest fixed point of $F$. Indeed, by Assumption \ref{it:Mmax}, for every  $x\in\Fix(F)$,  there is $y\in M(x)$ with $x\le y$.  Then $x\in A_M$ and $x\le x^*$. 
\end{proof}
\begin{rk}
	Assumption \ref{it:Minc} in Theorem \ref{thm:myCalciano} cannot be relaxed to upper increasingness of $M$. For example, let $X$ be as in Example \ref{eg:01abZhou} and define correspondences $F,M:X\to 2^X$ by $M(x)=F(x)=\{a,b\}$ for every $x\in X$. Then Assumptions \ref{it:MinF} and \ref{it:Mmax} are satisfied. The correspondence $M$ is upper and lower increasing. The subset $A_F$ is $\{\0,a,b\}$.  Then $x^*=\1$ which is not a fixed point of $F$.
\end{rk}
\begin{lm}\label{lm:7}Let $X$ be a nonempty poset. Let $M:X\to 2^X$ be a correspondence of nonempty values. Assume that for any $x<x'$ in $X$, every $y\in M(x)$ and every $y'\in M(x')$, one has $y\le y'$.
Let 	$S$ be a nonempty subset of $A_M$. If $z=\sup_XS$ exists,	then $z\in A_M$. \end{lm}
\begin{proof}Assume the contrary  $z\notin A_M$. Since $M(z)$ is nonempty, there is $y'\in M(z)$. For every $x\in S$, we have $x<z$. By definition of $A_M$, there is $y\in M(x)$ with $x\le y$.  By assumption, one has $y\le y'$. Thus,  $x\le y'$ for all $x\in S$. Then $z\le y'$. It implies $z\in A_M$, which is a contradiction. 
\end{proof}

Using Theorem \ref{thm:myCalciano}, we recover an extremal fixed point theorem due to  Calciano.
\begin{cor}[Calciano, {\cite[Theorem 1]{calciano2007games}, \cite[Theorem 10]{calciano2010theory}}]
	Let $X$ be a nonempty poset.  Let $F:X\to 2^X$ be an upper increasing correspondence. Suppose that 
	\begin{enumerate}[label=\emph{\arabic*})]
		\item for every $x\in X$, the value $F(x)$ is nonempty with a greatest element;
		
		\item the subset $A_F$ is nonempty;
		\item the element $x^*=\sup_X(A_F)$ exists.
	\end{enumerate}
	Then $x^*$ is the largest fixed point of $F$.
\end{cor}
\begin{proof}For every $x\in X$, let $M(x)$ to be the set of the greatest element of $F(x)$.
 We check Assumption \ref{it:Minc} in Theorem \ref{thm:myCalciano}. Since $F$ is upper increasing, for any $x<x'$ in $M$, every $y\in M(x)$ and every $y'\in M(x')$,  there is $b\in F(x')$ with $y\le b$. Since $y'$ is the largest element of $F(x')$, we have $b\le y'$. Thus, we obtain $y\le y'$. The result follows from Theorem \ref{thm:myCalciano}.
\end{proof}

Lemma \ref{lm:Zhou2} is used in one proof of Theorem \ref{thm:myZhou2}. If further  $F$ is ascending and every value is subcomplete in $X$, then Lemma \ref{lm:Zhou2} specializes to \cite[Theorem 2.5.1 (a)]{topkis1998supermodularity}.
\begin{lm}\label{lm:Zhou2}
	Let $X$ be a nonempty complete lattice. Let $F:X\to 2^X$ be an upper (resp. lower) C-ascending correspondence. Assume that  for every $x\in X$, the value $F(x)$ is nonempty and chain-bounded above (resp. below). 	Then $F$ has a largest  (resp. least) fixed point $\sup_X(A_F)$ (resp. $\inf_X(B_F)$).
\end{lm}

\begin{proof}By symmetry, it suffices to prove  the statement  without parentheses.
	
	For every $x\in X$, let $M(x)$ be the set of maximal elements of $F(x)$. We verify the conditions of Theorem \ref{thm:myCalciano}. Assumption \ref{it:MinF} holds. 
	
	For every $y\in F(x)$, consider $D:=\{z\in F(x):y\le z\}$. Then $y\in D$. Because $F(x)$ is chain-bounded above,  for every nonempty chain $C$ in $D$, there is $y'\in F(x)$ such that $C\le y'$. Hence, one obtains $y'\ge y$ and  $y'\in D$. By Zorn's lemma, $D$ has a maximal element $\bar{y}$. Then $\bar{y}\in M(x)$ and $y\le \bar{y}$. In particular, Assumption \ref{it:Mmax} is satisfied.
	
	Since the correspondence $F$ is upper C-ascending, for any $x<x'$ in $X$, every $y\in M(x)$ and every $y'\in M(x')$,  there is $y''\in F(x')$ with $y''\ge y\vee y'$. Now that  $y'$ is a maximal element of $F(x')$, we have $y'=y''\ge y$. Assumption \ref{it:Minc} is true.
	
	Because $X$ is complete, $\0=\min X$ exists. As $M(\0)$ is nonempty, one has  $\0\in A_F$. Thus,  Assumption \ref{it:Anonempty} is verified.
	By completeness of $X$ again, the element $x^*:=\sup_X (A_F)$ exists, so Assumption \ref{it:x*} is satisfied. Therefore, by Theorem \ref{thm:myCalciano}, the largest fixed point of $F$ is $x^*$.
\end{proof}
\section{Veinott's approach}\label{sec:1pf}
In the first approach to proving Theorem \ref{thm:myZhou2}, we need to show that certain ``truncation'' of the correspondence still has nonempty values. For this, we need 
Lemma \ref{lm:VeinottinZhou}. 
\begin{lm}\label{lm:VeinottinZhou}
	Let $(S,\le)$ be a poset. Let $P,Q$ be two subsets of $S$. Assume that $Q$ is chain-bounded below. Fix $x\in P\cap Q$. Suppose that for every $p\in P\setminus\{x\}$ and every $q\in Q$, there is $q'\in Q$ such that $q'\le p$ and $q'\le q$. Then there exists $q_0\in Q$ with $q_0\le P$.
\end{lm}
\begin{proof}
	The statement holds when $P=\{x\}$. Assume that $P\setminus\{x\}$ is nonempty. By the well ordering principle, there is a well order $\preceq$ on the nonempty set $P\setminus\{x\}$. Let $T$ be a copy of $P$ together with an exotic element $\infty$. Extend the order $\preceq$ from $P\setminus\{x\}$ to $T$ with $\min T=x$ and $\max T=\infty$. Then $(T,\preceq)$ is still a well ordered set. Define a map \[f:T\to P,\quad
	t\mapsto\begin{cases}
		t, & t\in P,\\
		x, & t=\infty.
	\end{cases}
	\]
	For every $t\in T$, let $P_t=\{f(t'):t'\preceq t\}$ and $Q_t:=\{q\in Q:q\le P_t\}$.
	By transfinite induction, we   prove that there is a family $(q_t)_{t\in T}$ with $q_t\in Q_t$ for all $t\in T$ and  $q_t\ge q_{t'}$ whenever $t\prec t'$ in $T$.
	
	For the base case $t=x$, one has $P_x=\{x\}$, then $Q_x$ is nonempty and has a largest element $q_x=x$.
	
	For $t_0(\neq x)\in T$,  assume that such a family $(q_t)_{t\prec t_0}$ with $q_x=x$ is constructed.
	As $Q$ is chain-bounded below, there exists $m\in Q$ such that $m\le q_t$ for every $t(\prec t_0)\in T$. In particular, one has $m\le x$. If $f(t_0)=x$, then $m\le f(t_0)$. We can take $q_{t_0}=m$. If $f(t_0)\neq x$,  by assumption there is  $q'\in Q$ with $q'\le m$ and $q'\le f(t_0)$. We can take $q_{t_0}=q'$. Since $P_{t_0}=\{f(t_0)\}\cup\cup_{t\prec t_0}P_t$, in both cases one has $q_{t_0}\in Q_{t_0}$ and $q_{t_0}\le q_t$ for every $t\prec t_0$. 
	As the induction is completed, the set $Q_{\infty}:=\{q\in Q:q\le P\}$ is nonempty.
\end{proof}

Now we are in a position to finish the proof of Theorem \ref{thm:myZhou2}.
\begin{proof}[First proof of Theorem \ref{thm:myZhou2}]By symmetry, it suffices to prove the statement in parentheses. For any $x\ge y$ in $X$, we prove that $F^y(x)$ is chain-bounded below. Indeed, for every nonempty chain $C$ inside, as $F(x)$ is chain-subcomplete downwards, $\inf_X(C)\in F(x)$. Since $y\le C$, one has $y\le \inf_X(C)$. Therefore, $\inf_X(C)$ is a lower bound on $C$ in $F^y(x)$.

	Therefore, by Lemma \ref{lm:Zhou2}, the poset $\Fix(F)$ is nonempty and has a least element. We show that $\Fix(F)$ is join-complete, i.e., $\sup_{\Fix(F)}(U)$ exists for every nonempty subset $U$ of $\Fix(F)$.
	
	By  completeness of $X$, the element $b:=\sup_XU$ exists. If $b\in \Fix(F)$, then $\sup_{\Fix(F)}(U)$ is $b$. Now assume  $b\notin \Fix(F)$. Because $F(b)$ is nonempty, one may fix $\beta\in F(b)$. Let $V:=\{\beta\}\cup U$. For every $u\in U$, we have $u\le b$. As $b\notin \Fix(F)$, we have a strict inequality $u<b$. For every $\beta'\in F(b)$, because $u\in F(u)$ and $F$ is upper C-ascending, there is $u'\in F(b)$ with $u'\ge u\vee \beta'$. Because $F(b)$ is chain-bounded above, by the dual of Lemma \ref{lm:VeinottinZhou}, there exists
	$w\in F(b)$ such that $w\ge V$. Since $V\supset U$, we have  $w\ge b$. From the first paragraph, every value of  $F^b:[b,\1]\to 2^{[b,\1]}$ is chain-bounded below. 
	
For every $s\in[b,\1]$,	we check   \begin{equation}\label{eq:F^bnon}
		F^b(s)\neq\emptyset.
	\end{equation} In fact, if $s=b$, then $w\in F^b(s)$. If $s>b$, since $F(s)$ is nonempty, there is $t\in F(s)$. Then there is $t'\in F(s)$ with $t'\ge t\vee w$. One obtains $b\le w\le t\vee w\le t'$ and hence $t'\in F^b(s)$.

	We prove that the correspondence $F^b$ is lower V-ascending. Indeed, for any $x<y$ in $[b,\1]$, every $s\in F^b(x)$ and every $t\in F^b(y)$, because $F$ is lower V-ascending, one has $s\wedge t\in F(x)$. Since $b\le s$ and $b\le t$, we have $b\le s\wedge t$ and $s\wedge t\in F^b(x)$.
	
	By Lemma \ref{lm:Zhou2}, the set $\Fix(F^b)=[b,\1]\cap \Fix(F)$ is nonempty and has a minimum element $m$. Then $\sup_{\Fix(F)}(U)=m$. 
	
	The proof is finished by \cite[Lemma 2.6]{yu2023generalization1}.
\end{proof}
\section{Calciano's approach}\label{sec:2pf}
In the second approach to Theorem \ref{thm:myZhou2}, we make use of ``truncation'' of the original correspondence by elements of $A_F$.  Compared with Veinott's approach given in Section \ref{sec:1pf}, a difference lies between the proofs of (\ref{eq:F^bnon}).

	Let	$X$ be a nonempty complete lattice. Let $F:X\to 2^X$ be a correspondence of nonempty values. Set $A'_F:=\{\sup_XS:S(\neq\emptyset)\subset \Fix(F)\}\cup\{\0\}$. 
\begin{thm}\label{thm:my13}

 Suppose that
	\begin{enumerate}[label=\emph{\roman*})]
		\item\label{it:Fchainbddabove} for every $x\in X$, the value $F(x)$ is chain-bounded above;	\item\label{it:Minc13} for any $x<x'$ in $X$, every maximal element $y\in F(x)$ and every maximal element $y'\in F(x')$, we have $y\le y'$;
		\item\label{it:lbinFh} for every $h\in A'_F$ and every $x\in [h,\1]$, the value $F^h(x)$ is chain-bounded below;
			\item \label{it:Lhinc}for every $h\in A'_F$, any $x<x'$ in $[h,\1]$, every minimal element $y\in F^h(x)$ and every minimal element $y'\in F^h(x')$, we have $y\le y'$.
	\end{enumerate}
	Then $\Fix(F)$ is a nonempty complete lattice. 
\end{thm}
\begin{proof}
	For every $h\in A_F$ and every $x\ge h$ in $X$,	we prove  \begin{equation}\label{eq:F^hnon}
		F^h(x)\neq\emptyset.
	\end{equation} 
	Indeed, by  definition  of $A_F$, there is $y\in F(h)$ with $h\le y$. If $x=h$, then $y\in F^h(x)$. Now assume $x>h$. From the  proof of Lemma \ref{lm:A'inA}, the set $F(x)$ (resp. $F(h)\cap [y,\1]$) has a maximal element $z$  (resp. $g$).  By Condition \ref{it:Minc13}, one has
	$z\ge g$, so $z\in F^h(x)$. Thus, \eqref{eq:F^hnon}  is proved.
	
By Lemma \ref{lm:A'inA} and \eqref{eq:F^hnon},	for every $h\in A'_F$ and every $x\ge h$ in $X$, the set $F^h(x)$ is nonempty.
 Let $L_h(x)$ be the set of minimal elements of $F^h(x)$.
	For every $y\in F^h(x)$, the set $F^h(x)\cap[\0,y]$ contains $y$, so it is nonempty. By  Condition \ref{it:lbinFh}, every nonempty chain $C$ in $F^h(x)\cap[\0,y]$ has a lower bound $l$ in $F^h(x)$.  Then $l\le C\le y$. Hence, one has $l\in F^h(x)\cap[\0,y]$. By Zorn's lemma, $F^h(x)\cap[\0,y]$ has a minimal element $z$. Then \begin{equation}\label{eq:manyminimal}z\in L_h(x),\quad z\le y.\end{equation} 
	
	For every $h\in A'_F$, we check the conditions of Theorem \ref{thm:myCalciano} for the correspondences $F^h,L_h:[h,\1]\to2^{[h,\1]}$. Assumption \ref{it:MinF} follows from (\ref{eq:F^hnon}).  Assumption \ref{it:Mmax} follows from (\ref{eq:manyminimal}). Assumption \ref{it:Minc} is exactly Condition \ref{it:Lhinc}. Assumption \ref{it:Anonempty} results from $\1\in B_{F^h}$. By  subcompleteness of the sublattice $[h,\1]$ in $X$, the element $x_h:=\inf_{[h,\1]}B_{F^h}=\inf_XB_{F^h}$ exists. Thus, 
Assumption \ref{it:x*} holds.

By Theorem  \ref{thm:myCalciano}, for every $h\in A'$, one has  $x_h=\min\Fix(F^h)$.	In particular, $x_{\0}$ is the least element  of $\Fix(F^{\0})=\Fix(F)$. Therefore,  $\Fix(F)$ is nonempty.	For a nonempty subset $S$ of $\Fix(F)$, let $s:=\sup_XS$. The least element of $\Fix(F^s)=\Fix(F)\cap [s,\1]=\{x\in\Fix(F):x\ge S\}$ is $x_s$. Hence, one has $x_s=\sup_{\Fix(F)}S$. 
	
	The proof is completed by \cite[Lemma 2.6]{yu2023generalization1}.
\end{proof}\begin{lm}\label{lm:A'inA}Assume that 	$F(x)$ is chain-bounded above for every $x\in X$,  and that for any $x<x'$ in $X$, every maximal element $y\in F(x)$ and every maximal element $y'\in F(x')$, we have $y\le y'$.
Then $A'_F\subset A_F$.
\end{lm}\begin{proof}
Since $F(\0)\neq\emptyset$, one has $\0\in A_F$. For every $x\in X$, let $M(x)$ be the set of maximal elements of $F(x)$. For every $y\in F(x)$, the set $F(x)\cap[y,\1]$ contains $y$, so  it is nonempty. By assumption, every nonempty chain $C$ in $F(x)\cap[y,\1]$ has an upper bound $s\in F(x)$. Since $s\ge C\ge y$, one has $s\in F(x)\cap[y,\1]$. By Zorn's lemma, $F(x)\cap[y,\1]$ has maximal element $g$. Then $g\in M(x)$ and $y\le g$. In particular, $M(x)$ is nonempty for every $x\in X$, and $A_F=A_M$. By  completeness of $X$, for every nonempty subset $S$ of $\Fix(F)$,  the supremum $\sup_XS$ exists. By assumption, Lemma \ref{lm:7} and $S\subset A_F=A_M$, one has $\sup_XS\in A_F$.
\end{proof}
\begin{proof}[Second proof of Theorem \ref{thm:myZhou2}]
	We check the conditions of Theorem \ref{thm:my13}. Condition \ref{it:Fchainbddabove} is Assumption \ref{it:Fischainbddabove}.
	
By Assumption \ref{it:valuehalfsubcpltZhou},	for any $h\le x$ in $X$ and every nonempty chain $C$ in $F^h(x)$,  one has $\inf_X(C)\in F(x)$. Further, $h\le C$ implies $h\le \inf_X(C)$. Thus, $\inf_X(C)\in F^h(x)$ is a lower bound on $C$. In particular, $F^h(x)$ is chain-bounded below, i.e., Condition \ref{it:lbinFh} is verified.
	
Because $F$ is upper C-ascending, 	for any $x<x'$ in $X$, every maximal element $y\in F(x)$ and every maximal element $y'\in F(x')$, there exists $u\in F(x')$ with $y\vee y'\le u$. As $y'$ is maximal in $F(x')$, we have $y'=u\ge y$. Condition \ref{it:Minc13} is proved.
	
Because $F$ is lower V-ascending,	for any $h\le x<x'$ in $X$, every minimal element $y\in F^h(x)$ and every minimal element $y'\in F^h(x')$,  one gets $y\wedge y'\in F(x)$. Since $h\le y$ and $h\le y'$, we have $h\le y\wedge y'$ and hence $y\wedge y'\in F^h(x)$. As $y$ is minimal in $F^h(x)$, we have $y=y\wedge y'\le y'$. Condition \ref{it:Lhinc} is obtained.
	
	Now, the conclusion follows from Theorem \ref{thm:my13}.
\end{proof}
As applications of Theorem \ref{thm:my13}, we recover Calciano's and Sabarwal's fixed point theorems.
\begin{proof}[Proof of Theorem \ref{thm:CalZhou}]
	We check the conditions of Theorem \ref{thm:my13}.	For every $x\in X$, every nonempty chain in $F(x)$ has an upper bound, namely $\max F(x)$ in Assumption \ref{it:Fhasmax}. Thus, Condition \ref{it:Fchainbddabove} holds.
	
By Assumption \ref{it:Fuppinc} and \cite[Theorem 11]{calciano2010theory}, for any $x<x'$ in $X$, every maximal element $y\in F(x)$ and every maximal element $y'\in F(x')$,  we have $y\le y'$. Condition \ref{it:Minc13} is verified. Then by Lemma \ref{lm:A'inA},  every $h\in A'_F$ is in $A_F$. Whence, by Assumption \ref{it:Fhmin}, $\min F^h(x)$ exists  for every $x\in [h,\1]$. It is a  lower bound on every nonempty chain in $F^h(x)$. Thereby, Condition \ref{it:lbinFh} holds.

By Assumption \ref{it:Fstr}, for any $x<x'$ in $[h,\1]$, every minimal element $y\in F^h(x)$ and every minimal element $y'\in F^h(x')$, there exists $p\in F(x)$ with $y\wedge y'\le p\le y'$.  Since $h\le y\wedge y'$, one has $p\in F^h(x)$. From Assumption \ref{it:Fhmin}, the least element of $F^h(x)$ is $y$, so $y\le p\le y'$. Condition \ref{it:Lhinc} follows. We conclude by Theorem \ref{thm:my13}.
\end{proof}

\begin{proof}[Proof of Theorem \ref{thm:Sabarwal}]
	We check the conditions of Theorem \ref{thm:my13}. By Assumption \ref{it:isotonesup}, Conditions \ref{it:Fchainbddabove} and \ref{it:Minc13} hold. Then by Lemma \ref{lm:A'inA} and Assumption \ref{it:isoinfupint},   Conditions \ref{it:lbinFh} and \ref{it:Lhinc}  are true. 
\end{proof}

\begin{proof}[Proof of Theorem \ref{thm:chaincpltvalue}]
	It suffices to check the conditions of Theorem \ref{thm:my13}. 
	
	For every $h\in A'_F$ and every $x\in[h,\1]$, we prove that every nonempty chain $C$ has a lower bound in $F^h(x)$. Assume the contrary. Because $F(x)$ is chain-complete downwards, one has $h\neq\0$. By definition of $A'_F$, there is a nonempty subset $S$ of $\Fix(F)$ with $h=\sup_XS$. For every $s\in S$, one has $s\le h\le x$. If $s=x$, then $x=h=s\in\Fix(F)$ and $h\in F^h(x)$ is a lower bound on $C$, which is a contradiction. Now assume $s<x$.  As $F(x)$ is chain-complete downwards, $m:=\inf_{F(x)}C$ exists. Since $F$ is upper V-ascending,  $s\vee m$ is in $F(x)$. From $C\ge m$ and  $C\ge h\ge s$, one obtains $C\ge s\vee m$. It implies $m\ge s\vee m$. Thus, one has $S\le m$ and hence $h\le m$. However, $m\in F^h(x)$ is a lower bound on $C$, which is a contradiction. Thus, Condition \ref{it:lbinFh} is verified.

Because $F$ is lower V-ascending, 	for every $h\in X$, any $x<x'$ in $[h,\1]$, any minimal elements $y\in F^h(x)$ and  $y'\in F^h(x')$, one has $y\wedge y'\in F(x)$ and $h\le y\wedge y'$. Then $y\wedge y'\in F^h(x)$. Because $y$ is minimal, $y=y\wedge y'\le y'$. Condition  \ref{it:Lhinc} is proved.
	
	We check Condition \ref{it:Minc13}.  Because $F$ is upper V-ascending, 	  for any $x<x'$ in $X$, any maximal elements $y\in F(x)$ and $y'\in F(x')$, one has $y\vee y'\in F(x')$. Because $y'\le y\vee y'$ and $y'$ is maximal, one has $y\le y\vee y'=y'$.
	\end{proof}
\begin{eg}\label{eg:Zhou>wkZhou}
	Let $X$ be $[-1,1]\cup\{a,b\}\cup \{2\}$, where $1\le a,b\le 2$ and $a,b$ are incomparable. Then $X$ is complete lattice. Define a correspondence \[F:X\to 2^X,\quad x\mapsto\begin{cases}
		\{-1\}\cup(-0.5,0.5)\cup\{1\}, & x\le 1,\\
		\{a,b\}, & x=a,\\
		\{a\}, & x=b,\\
		\{2\}, & x=2.
	\end{cases}\] The correspondence $F$ is  not ascending. Still,  it satisfies the hypotheses of Theorem \ref{thm:chaincpltvalue}. The value $F(0)$ is neither chain-subcomplete upwards nor chain-subcomplete downwards in $X$. The value $F(a)$ is not a lattice. The set $\Fix(F)=\{-1\}\cup(-0.5,0.5)\cup\{1,a,2\} $ is indeed a complete lattice. 
\end{eg}
\section{Multivalued fixed point theorems for posets}\label{sec:ABM}
If $F$ is a single-valued map, then Theorem \ref{thm:AbianBrown} (resp. \ref{thm:Markowski}) specializes to the Abian-Brown theorem (resp. the Markowsky theorem), i.e., Fact \ref{ft:AbianBrown} (resp.  \ref{ft:Markowsky}).
\begin{thm}\label{thm:AbianBrown}
	Let $X$ be a nonempty poset where  every nonempty well-ordered subset  has a least upper bound. Let $F:X\to 2^X$ be a correspondence. Assume that \begin{enumerate}\item\label{it:FincAB} for every nonempty well-ordered subset  $C$ of $X$, every $x\in X$ with $x>C$
	and every selection $f:C\to X$ of $F$, there is $y\in F(x)$ with $y\ge f(C)$; \item the set $A_F$ is nonempty.\end{enumerate}Then $F$ has a fixed point.  
\end{thm}\begin{proof}
 For every nonempty well-ordered subset $C$ of $A_F$, we prove that $x:=\sup_XC$ is in $A_F$.  Assume to the contrary that $x\notin A_F$. Then $x>C$. For every $c\in C$, there is $f(c)\in F(c)$ with $c\le f(c)$. Then $f:C\to X$ is a selection of $F$. By Assumption \ref{it:FincAB}, there is $y\in F(x)$ with $y\ge f(C)$. Then $y\ge C$, so $y\ge x$. Thus, one has $x\in A_F$, which is a contradiction.

Assume to the contrary that $F$ has no fixed point. For every $x\in A_F$, there is $g(x)\in F(x)$ with $x\le g(x)$. Since $x\notin \Fix(F)$, one has $x<g(x)$. By Assumption \ref{it:FincAB}, there is $z\in F(g(x))$ with $z\ge g(x)$. Hence, one has $g(x)\in A_F$. By Moroianu's fixed point theorem \cite[Theorem]{moroianu1982theorem}, since every nonempty well-ordered subset of $A_F$ has an upper bound in $A_F$,  the  map $g:A_F\to A_F$ has a fixed point. As $g$ is a selection of $F$, it is a contradiction. 
\end{proof}

\begin{thm}\label{thm:Markowski}
Let $X$ be a nonempty  poset. Let $F:X\to 2^X$ be a correspondence. Assume that\begin{enumerate}[label=\alph*)]
	\item\label{it:Markowskichaincplt} $X$ is chain-complete;
\item\label{it:stronginc} for every chain $U$ in $X$, every subset $V$ of $X$, any selections $f:U\to X$ and $g:V\to X$ of $F$ with $f(U)\le g(V)$, and every $x\in X$ with $U< x< V$, there is $y\in F(x)$ with $f(U)\le y\le g(V)$. 
\end{enumerate}  Then
\begin{enumerate}
\item\label{it:multivalleast} there is a least element  of $\Fix(F)$;
\item one has $\min \Fix(F)\le B_F$;
\item $\Fix(F)$ is chain-complete.
\end{enumerate}
\end{thm}
\begin{proof}\begin{enumerate}
		\item From Assumption \ref{it:Markowskichaincplt}, there is a least element $\0$ of $X$. If $\0\in\Fix(F)$, then $\min\Fix(F)=\0$. Now assume $\0\notin\Fix(F)$, then $0<\Fix(F)$. By Assumption \ref{it:stronginc},  there is $y_0\in F(\0)$ with $y_0\le \Fix(F)$. Let $S$ be the subset of $x\in X$ for which there is $y\in F(x)$ with $x\le y\le \Fix(F)$. Then $S$ is nonempty as $\0\in S$.
		
	For every nonempty chain $C$ in $S$, we prove that $x:=\sup_XC$ is in $S$. Assume the contrary $x\notin S$. Since $S\le\Fix(F)$, one has $x\le \Fix(F)$. From $x\notin S$, one has $x\notin\Fix(F)$ and hence $C<x<\Fix(F)$. For every $c\in C$, there is $f(c)\in F(c)$ with $c\le f(c)\le \Fix(F)$. Then $f:C\to X$ is a selection of $F$. By Assumption \ref{it:stronginc}, there is $y\in F(x)$ with $f(C)\le y\le \Fix(F)$. Then $C\le y$. One has $x\le y$ and hence $x\in S$, which is a contradiction. 

 Assume the contrary that $\Fix(F)$ has no least element. For every $x\in S$, there is $g(x)\in F(x)$ with $x\le g(x)\le \Fix(F)$. Every element of $S\cap\Fix(F)$ is a least element of $\Fix(F)$, so $x\notin \Fix(F)$ and hence $x<g(x)$. Since $g(x)$ is not the least element of $\Fix(F)$, one has $g(x)<\Fix(F)$. Then by Assumption \ref{it:stronginc},  there is $z\in F(g(x))$ with $g(x)\le z\le \Fix(F)$. Then $g(x)\in S$. By the Bourbaki–Witt fixed point theorem (see, e.g., \cite[Thm.~2.1, p.881]{Lang2002algebra}), the  map $g:S\to S$ has a fixed point $x_0\in S$. Since $g$ is a selection of $F$, one has $x_0\in S\cap \Fix(F)$ and hence $x_0=\min \Fix(F)$, which is a contradiction. 
\item For every $v\in B_F$, there is $w\in F(v)$ with $w\le v$. Define a correspondence $F':[\0,v]\to 2^{[\0,w]},\quad x\mapsto F(x)\cap [\0,w]$. For every chain $U\subset [\0,v]$, every subset $V\subset [\0,v]$, any selections $f:U\to [\0,v]$ and $g:V\to [\0,v]$ of $F'$ with $f(U)\le g(V)$, and every $x\in [\0,v]$ with $U<x<V$, we shall find $y\in F'(x)$ with $f(U)\le y\le g(V)$. We divide the situation to three cases.\begin{enumerate}[label=\roman*]
\item\label{it:VnonemptyZhou}  If $V$ is nonempty, then  by Assumption \ref{it:stronginc}, there is $y\in F(x)$ with $f(U)\le y\le g(V)$. Since $g(V)$  is a nonempty subset of $[\0,w]$, one has $y\le w$ and hence $y\in F'(x)$. 
\item If $V$ is empty and $x<v$, then by setting $V=\{v\}$ and $g(v)=w$ we are reduced to Case \ref{it:VnonemptyZhou}.
\item If $x=v$, then take $y=w$.
\end{enumerate}

 The interval $[\0,v]$ is  chain-complete, and $F':[\0,v]\to 2^{[\0,v]}$ satisfies Assumption \ref{it:stronginc}, so by Part \ref{it:multivalleast}, there is $v'\in\Fix(F')$. Since $\Fix(F')\subset\Fix(F)$, one has $\min \Fix(F)\le v'\le v$. 
\item By Part \ref{it:multivalleast}, it remains to prove that for every nonempty chain $C$ in $\Fix(F)$, the element $\sup_{\Fix(F)}C$ exists. Let $u:=\sup_XC$. If $u\in C$, then $\sup_{\Fix(F)}C$ is $u$. Now assume $u\notin C$. One has $C<u$. The poset $[u,+\infty)$	is   chain-complete.  For every chain $U$ in $[u,+\infty)$, every subset $V$ of $[u,+\infty)$, any selections $f:U\to [u,+\infty)$ and $g:V\to [u,+\infty)$ of $F^u$ with $f(U)\le g(V)$, and every $x\in [u,+\infty)$ with $U< x< V$, we shall find $y\in F^u(x)$ with $f(U)\le y\le g(V)$. There are two cases.\begin{itemize}
\item If $U$ is nonempty, then  by Assumption \ref{it:stronginc}, there is $y\in F(x)$ with $f(U)\le y\le g(V)$. Since $f(U)$ is a nonempty subset of $[u,+\infty)$, one has $u\le y$ and hence $y\in F^u(x)$.
\item If $U$ is empty, then by $C<u\le x<V$ and Assumption \ref{it:stronginc}, there is $y\in F(x)$ with $C\le y\le g(V)$. One has $u\le y$ and hence $y\in F^u(x)$.
\end{itemize}

Now by Part \ref{it:multivalleast}, there is a least element $u'$ of $\Fix(F^u)$. Since $\Fix(F^u)=\Fix(F)\cap [u,+\infty)=\{x\in\Fix(F):x\ge C\}$, one has $u'=\sup_{\Fix(F)}C$.\end{enumerate}
\end{proof}\begin{rk}
Let $X$ be a lattice. Let $F:X\to 2^X$ be a correspondence. 
\begin{enumerate}\item If $F$ is V-ascending  and of nonempty subcomplete values, then  for any subsets $U$ and $V$  of $X$,  any selections $f:U\to X$ and $g:V\to X$ of $F$ with $f(U)\le g(V)$, and every $x\in X$ with $U< x< V$, there is $y\in F(x)$ with $f(U)\le y\le g(V)$. In particular, $F$ 	satisfies Assumption \ref{it:stronginc} in Theorem \ref{thm:Markowski}. In fact, since $f(x)$ is nonempty, there is $z\in f(x)$. By ascendingness, for every $u\in U$, one has $f(u)\vee z\in f(x)$. Since $f(x)\subset X$ is subcomplete, $z':=\sup_X\{f(u)\vee z|u\in U\}$ is in $f(x)$.  Replacing $z$ by $z'$, one may assume $z\ge f(U)$. By ascendingness, for every $v\in V$, one has $g(v)\wedge z\in f(x)$. Since $f(x)\subset X$ is subcomplete, $y:=\inf_X\{g(v)\wedge z|v\in V\}$ is in $f(x)$. For every $v\in V$, one has $f(U)\le g(v)$ and $f(U)\le z$, so $f(U)\le g(v)\wedge z\le g(v)$. Thus, $f(U)\le y\le g(V)$. 
	\item Since every well-ordered subset is a chain, Assumption \ref{it:stronginc} in Theorem \ref{thm:Markowski} is stronger than Assumption \ref{it:FincAB} in Theorem \ref{thm:AbianBrown}.\end{enumerate}
\end{rk}
\section{Application to partially quasi-supermodular games}\label{sec:game}
We provide an application of Theorem \ref{thm:myZhou2} to partially quasi-supermodular games (Definition \ref{df:parqG}), which are extensions of quasi-supermodular games in \cite[p.175]{milgrom1994monotone} and \cite[p.179]{topkis1998supermodularity}.

\begin{df}
Let $X$ be a lattice, $Y$ be a chain and $f: X \to Y$ be a map. We say that $f$ is \emph{partially quasi-supermodular} if it satisfies the following properties.  For any $x,x' \in X$, if $f(x') \ge f(x \vee x')$ (resp. $f(x') > f(x \vee x')$), then there is $u \in X$, such that $u\le x\wedge x'$ and $f(u) \ge f(x)$ (resp. $f(u) > f(x)$). 
\end{df}
A quasi-supermodular function in \cite[p.162]{milgrom1994monotone} is partially quasi-supermodular. 
\begin{df}
Let $X$ be a lattice, $Y$ be a chain and $f: X \to Y$ be a map. If for any $x,x'\in X$, either $f(x\wedge x')\ge f(x)\wedge f(x')$ or $f(x\vee x')\ge f(x)$, then $f$ is called join-superextremal.
\end{df}
 As shown in Example \ref{eg:game}, a partially quasi-supermodular and join-superextremal function is not necessarily quasi-supermodular.

\begin{df}[{\cite[p.1260]{milgrom1990rationalizability}}, {\cite[Definition 1]{prokopovych2017strategic}}]\label{df:halforderusc}
		Given a complete lattice $X$, a function $f:X\to \R$ is upward upper semicontinuous if for every nonempty chain $C$ in $X$,  one has \[\limsup_{x\in C,x\to\sup_X C}f(x)\le f(\sup_X C).\]
		
\end{df}

We introduce our extension of quasi-supermodular games.
\begin{df}\label{df:parqG}
    A partially quasi-supermodular game $(N,\{S_i\}_i,\{f_i\}_i)$ consists of the following data:

    \begin{itemize}
        \item a nonempty set of players $N$;
        \item for every $i\in N$, a nonempty lattice $S_i$ of the strategies of player $i$;
        \item for every $i\in N$, a payoff function $f_i: S:=\prod_i S_i \to \R$
    \end{itemize}
    such that for every $i \in N$,
    \begin{enumerate}
        \item for every $x_{-i} \in S_{-i}:= \prod_{j\neq i} S_j$, the function $f_i(\cdot, x_{-i}): S_i \to \R$ is partially quasi-supermodular and 
        \item the function $f_i: S_i \times S_{-i} \to \R$ has the single crossing property in the sense of \cite[p.160]{milgrom1994monotone}.
      
    \end{enumerate}
\end{df}
Fix a partially quasi-supermodular game.

\begin{df}
		A strategy $x\in S$ is called a \emph{Nash equilibrium} of the fixed partially quasi-supermodular game, if for every $i\in N$ and every $x_i'\in S_i$, one has $f_i(x'_i,x_{-i})\le f_i(x)$.
\end{df}
  For every $i\in N$,  define a correspondence $ R_i: S \to 2^{S_i},\, x \mapsto \argmax_{y_i\in S_i} f_i(y_i,x_{-i})$. We define  $R: S\to 2^S,\quad x \mapsto \prod_{j\in N} R_j(x)$. Then $\Fix(R)$ is exactly the set of Nash equilibria. 
  
 The topology on a poset $X$ generated by the subbasis for closed subsets  comprised of closed intervals $(-\infty,x]$ and $[x,+\infty)$ (where $x$ runs over $X$)  is called the \emph{interval topology} of $X$.

\begin{thm}\label{thm:game}
     Suppose that the $S_i$ are complete.  
    Assume that for every $i\in N$ and every $x \in S$,  
    \begin{enumerate}

    \item \label{ass:b} for every nonempty chain $C \subset S_i$, there is $b \in S_i$ such that $b\le C$ and $f_i(b,x_{-i}) \ge \liminf_{c \in C, c\to \inf_{S_i} C} f_i(c,x_{-i})$;
    \item \label{ass:continuous} the function $f_i(\cdot, x_{-i}): S_i \to \R$ is either join superextremal and upward upper semi-continuous, or upper semi-continuous with respect to the interval topology of $S_i$.
    \end{enumerate}

    Then the set of Nash equilibria is a nonempty complete lattice. 
\end{thm}
\begin{proof}
  The sets $S$ and $S_{-i}$ are complete lattices.  We prove that for every $i\in N$, every $x\in S$ and every $a\in \R$, the set $L_i^{a}:=\{y_i\in S_i: f_i(y_i,x_{-i}) \ge a\}$ is chain-subcomplete upwards in $S_i$. Indeed, by \cite[Prop. 3.6]{yu2023topkis2},  Assumption \ref{ass:continuous} implies that $f_i(\cdot,x_{-i})$ is upward upper semicontinuous. So for every nonempty chain $C$ in $L_i^{a}$, one has  \[f_i(\sup_{S_i} C, x_{-i}) \ge \limsup_{y_i\in C,y_i\to \sup_{S_i}C} f_i(y_i,x_{-i}) \ge a.\] Then we get $\sup_{S_i} C \in L_i^{a}$. Therefore, $L_i^{a}$ is chain-subcomplete upwards in $S_i$. Take $b$ as in Assumption \ref{ass:b}. Then \[f_i(b,x_{-i}) \ge \liminf_{c\in C,c\to \inf_S C} f_i(c,x_{-i}) \ge a.\] 
  The lower bound $b$ on $C$ lies in $L_i^{a}$.   Therefore, $R(x)= \prod_{i\in N} L_i^{\max_{S_i} f_i(\cdot,x_{-i})}$ is chain-subcomplete upwards in $S$ and chain bounded below.

We show that $R_i(x)$ is nonempty in both cases of Assumption \ref{ass:continuous}. \begin{itemize}
    \item Assume that  $f_i(\cdot, x_{-i})$ is  join superextremal and upward upper semi-continuous. By Zorn's lemma, $L_i^a$ has a minimal element. Then from \cite[Cor. 3.4(3)]{yu2023topkis3normal},  $R_i(x)$ is nonempty.
    \item Assume that $f_i(\cdot, x_{-i})$  is upper semi-continuous with respect to the interval topology. By \cite[Theorem 20, p.250]{birkhoff1940lattice}, the interval topology of $S_i$ is compact. Then from \cite[Theorem 3, p.361]{bourbaki2013general}, $R_i(x)$ is nonempty.
\end{itemize}

    We check that $R$ is lower C-ascending. For any $x\le x'$ in $S$, every $i\in N$, every $y_i\in R_i(x)$ and every $y'_i \in R_i(x')$, one has \[f_i(y'_i,x'_{-i})\ge f_i(y_i\vee y'_i,x'_{-i}).\] By single crossing property, one has \[f_i(y_i',x_{-i}) \ge f_i(y_i\vee y_i',x_{-i}).\] Then by partial quasi-supermodularity of $f_i(\cdot,x_{-i})$, there is $u_i \in S_i$ such that $u_i \le y_i\wedge y_i'$ and $f_i(u_i,x_{-i}) \ge f_i(y_i,x_{-i})$. As a consequence, one has $u_i \in R_i(x)$. Set $u:=(u_i)_{i\in N}$. Then $u \in R(x)$ and $u\le y \wedge y'$.

    We check that $R$ is upper V-ascending. For any $x \le x'$ in $S$, every $i\in N$, every $y_i\in R_i(x)$ and every $y'_i \in R_i(x')$, every $u_i \in S_i$ with $u_i \le y_i \wedge y_i'$, one has \[f_i(y_i,x_{-i}) \ge f_i(u_i,x_{-i}).\] Therefore, by partial qulasi-supermodularity of $f_i(\cdot,x_{-i})$, one has \[f_i(y_i',x_{-i}) \le f_i(y_i \vee y_i',x_{-i}).\]  By single crossing property, one has $f_i(y_i',x'_{-i}) \le f_i(y_i\vee y_i',x'_{-i})$, which implies $y_i\vee y_i' \in R_i(x')$. Then $y\vee y' \in R(x')$. 

    By Theorem \ref{thm:myZhou2}, the poset $\Fix (R)$ is a nonempty complete lattice. 

\end{proof}

We give a specific example that satisfies the conditions of Theorem \ref{thm:game}. 
\begin{eg}\label{eg:game}
    We define a two-player normal form game as follows. Set \[S_1=\{(0,0),(1,1),(1,2),(2,1),(2,2)\},\] which is a subcomplete sublattice of $\R^2$ and $S_2=\{0\}$. The payoff function $f_i: S \to \R$ ($i=1,2$) is defined by $f_2\equiv 0$ and 
    \[f_1(x)= \begin{cases}
        10& \text{if } x_1=(0,0),\\ 0& \text{if } x_1=(1,1),\\1& \text{else}.
    \end{cases}\]
    The function $f_1(\cdot, 0): S_1 \to \R$ is partially quasi-supermodular and join-superextremal but not quasi-supermodular. Moverover, $f_1$ has the single crossing property. Therefore, this game is a partially quasi-supermodular game.  Since every chain in $S_i$ is finite for $i=1,2$, the assumptions in Theorem \ref{thm:game} are satisfied. The unique Nash equilibrium is $((0,0),0)$.
\end{eg}

	\subsection*{Acknowledgments}I thank my supervisor, Professor Philippe Bich, for his constant support. I am grateful to Professor Amir Rabah and Professor Łukasz Woźny for their valuable suggestions and encouragement. \printbibliography
	
\end{document}